\documentclass[journal,twoside,10pt]{IEEEtran} 

\usepackage{cite}

 \ifCLASSINFOpdf
   \usepackage[pdftex]{graphicx}
 \else
  \usepackage[dvips]{graphicx}
 \fi
\usepackage[cmex10]{amsmath}
\usepackage{algorithmicx}
\usepackage{algpseudocode}

\usepackage{tabularx}

  \usepackage[caption=false,font=footnotesize]{subfig}
\usepackage{url}

\usepackage{mathtools}
\usepackage{amssymb}
\usepackage{bbm}
\usepackage{hyperref}
\usepackage[capitalize]{cleveref}
\crefname{equation}{}{}
\usepackage{longtable}
\usepackage{bm}
\usepackage{amsthm}
\usepackage{commath}
\usepackage{framed}
\usepackage{etoolbox}
\newcounter{tempEquationCounter} 
\newcounter{thisEquationNumber}

\newtheorem{lemma}{Lemma}[section]
\newtheorem{proposition}{Proposition}[section]
\theoremstyle{definition}
\begin{document}
%
\title{Fronthaul-Limited Uplink OFDMA in Ultra-Dense CRAN with Hybrid Decoding}
%
%
%

\author{Reuben~George~Stephen,~\IEEEmembership{Student Member,~IEEE}, and~Rui~Zhang,~\IEEEmembership{Fellow,~IEEE}
\thanks{Copyright \copyright 2017 IEEE. Personal use of this material is permitted. However, permission to use this material for any other purposes must be obtained from the IEEE by sending a request to pubs-permissions@ieee.org.}
\thanks{This work was supported in part by the National University of Singapore under Research Grant R-263-000-B46-112.}
\thanks{This paper was presented in part at the IEEE Global Communications Conference~(GLOBECOM), Washington D.C., USA, Dec.\ 4--8, 2016.}
\thanks{R. G. Stephen is with the National University of Singapore Graduate School for Integrative Sciences and Engineering, National University of Singapore, Singapore 117456, and also with the Department of Electrical and Computer Engineering, National University of Singapore, Singapore 117583 (e-mail: reubenstephen@u.nus.edu).}
\thanks{R. Zhang is with the Department of Electrical and Computer Engineering, National University of Singapore, Singapore 117583 
(e-mail: elezhang@nus.edu.sg).}
}
%
%
 \markboth{IEEE Transactions on Vehicular Technology}%
 {Accepted for Publication}
%
\maketitle
\begin{abstract}
In an ultra-dense cloud radio access network~(UD-CRAN), a large number of remote radio heads~(RRHs), typically employed as simple relay nodes, are distributed in a target area, which could even outnumber their served users. However, one major challenge is that the large amount of information required to be transferred between each RRH and the central processor~(CP) for joint signal processing can easily exceed the capacity of the fronthaul links connecting them. This motivates our study in this paper on a new \emph{hybrid} decoding scheme where in addition to quantizing and forwarding the received signals for joint decoding at the CP, i.e. \emph{forward-and-decode~(FaD)} as in the conventional CRAN, the RRHs can locally \emph{decode-and-forward~(DaF)} the user messages to save the fronthaul capacity. In particular, we consider the uplink transmission in an orthogonal frequency division multiple access~(OFDMA)-based UD-CRAN, where the proposed hybrid decoding is performed on each OFDMA sub-channel~(SC). We study a joint optimization of the processing mode selections~(DaF or FaD), user-SC assignments 
and the users' transmit power allocations over all SCs to maximize their weighted-sum-rate subject to the RRHs' individual fronthaul capacity constraints and the users' individual power constraints. Although the problem is non-convex, we propose a Lagrange duality based solution, which can be efficiently computed with good accuracy. Further, we propose a low-complexity greedy algorithm which is shown to achieve close to the optimal performance under practical setups. Simulation results show the promising throughput gains of the proposed designs with hybrid decoding, compared to the existing schemes that perform either DaF or FaD at all SCs. 
\end{abstract}
\begin{IEEEkeywords}
Cloud radio access network, hybrid decoding, orthogonal frequency division multiple access, resource allocation, ultra-dense network.
\end{IEEEkeywords}
\section{Introduction}
Increasing the number of cellular base stations~(BSs) to serve a given area, also known as network densification, is foreseen to be a necessary solution 
to address the large data rate 
demands of the future fifth-generation~(5G) wireless communication networks~\cite{andrews-etal2014what,bhushan-etal2014network}. Cloud radio access network~(CRAN) provides a cost-effective way to achieve network densification, by replacing the conventional BSs with low-power 
distributed remote radio heads~(RRHs) that are deployed close to the users and coordinated by a central processor~(CP)~\cite{andrews-etal2014what}. CRAN significantly improves both the spectral efficiency and energy efficiency compared to conventional cellular networks, due to the centralized resource allocation and joint signal processing over the RRHs at the CP~\cite{park-etal2013joint,zhou-yu2014optimized,liu-zhang2015optimized,jain-etal2016backhaul,zhuang-lau2014backhaul,dai-yu2014sparse,shi-etal2014group,fan-etal2016randomized,luo-etal2015downlink,tao-etal2016content,shi-etal2016smoothed,fan-etal2016dynamic,peng-etal2016recent,liu-etal2015joint,stephen-zhang2016green}. Combining a dense network with centralized joint processing as in the CRAN leads to a powerful new network architecture termed ultra-dense CRAN~(UD-CRAN)~\cite{shi-etal2015large,stephen-zhang2017joint}, in which the number of RRHs can even exceed the number of users being served in a given area to support wireless connectivity of ultra-high throughput. 

In a UD-CRAN, the RRHs exchange data and control signals with the CP via high-speed wired or wireless links which are referred to as the fronthaul. With a fully centralized architecture, where the CP performs all the encoding/decoding operations, the RRHs in a UD-CRAN can be simple relay nodes that transmit or receive quantized/compressed baseband signals over the fronthaul links~\cite{park-etal2013joint,zhou-yu2014optimized,liu-etal2015joint,liu-zhang2015optimized,jain-etal2016backhaul}. Although such an architecture provides the maximum joint signal processing gains, the fronthaul links can get saturated easily by the large volume of signals that need to be transmitted over them, especially in UD networks. On the other hand, in traditional cellular networks, the BSs themselves have encoding/decoding capability, and only the user messages are sent over the backhaul, which typically requires much less capacity compared to sending quantized signals over the fronthaul in CRAN. However, in such networks, the performance gains 
due to joint signal processing are compromised. This motivates an alternative \emph{hybrid} architecture for UD-CRAN as proposed in~\cite{bi-etal2015wireless}, by which the benefits of both the BS-centric local processing in conventional cellular networks and the centralized processing in CRAN can be achieved by adaptively switching between the two, based on e.g., the wireless channel conditions, the user rate requirements and the fronthaul constraints of the RRHs. 

In this paper, we consider the uplink transmission in an orthogonal frequency division multiple access~(OFDMA)-based UD-CRAN with hybrid decoding on each sub-channel~(SC), as shown in~\cref{F:SysModel}. On each SC, the RRHs can choose either to: decode the assigned user's message locally and then forward it to the CP~(namely, decode-and-forward~(DaF)), or simply quantize the signal and forward it to the CP for joint decoding~(termed forward-and-decode~(FaD)), or not process the received signal at all~(to save fronthaul capacity). Most of the existing work on CRAN considers single-channel systems~\cite{park-etal2013joint,zhou-yu2014optimized,liu-zhang2015optimized,jain-etal2016backhaul,zhuang-lau2014backhaul,dai-yu2014sparse,shi-etal2014group,fan-etal2016randomized,luo-etal2015downlink,tao-etal2016content,shi-etal2016smoothed,fan-etal2016dynamic,peng-etal2016recent}, where all the users transmit in the same bandwidth. Previously, a hybrid compression and message-sharing strategy was considered in~\cite{patil-yu2014hybrid} for the downlink transmission in a single-channel CRAN, where the CP shares the uncompressed messages of some users to a subset of RRHs, in addition to transmitting compressed messages of all users to each RRH. In contrast, in this paper we consider the OFDMA-based UD-CRAN with multiple SCs for the uplink transmission with hybrid decoding on each SC. Although allowing multiple users to transmit on the same SC can potentially increase the system throughput, it would require the more complex successive interference cancellation receiver on each SC at the CP and/or the RRHs, compared to the simple single-user decoder considered in this paper. The uplink transmission in an OFDMA-based CRAN, 
with only FaD processing on each SC was considered in~\cite{liu-etal2015joint}. In contrast, with the hybrid decoding considered in this paper, each RRH can choose to first decode the message in an SC locally, before forwarding it to the CP, in addition to quantizing and forwarding the signal to the CP as in~\cite{liu-etal2015joint}. For the downlink of a multi-hop fronthaul CRAN, it has been shown in~\cite{liu-yu2017cross} that multicasting user messages with the help of network coding can perform better than unicasting compressed messages with simple routing. 

Specifically, we consider a cluster of $M$ RRHs and $K$ users, both with single-antennas, in an OFDMA-based UD-CRAN with $N$ orthogonal SCs, 
as shown in~\cref{F:SysModel}. The signal transmitted by the user assigned to each SC can be either quantized and forwarded by a subset of RRHs for joint decoding at the CP, or instead, locally decoded and forwarded by a single RRH. Jointly decoding the message of a user assigned to any SC at the CP provides a receive-combining gain from multiple RRHs over the local decoding 
at a single RRH, and hence achieves a higher transmission rate for the user in general. However, to achieve this rate improvement, 
the participating RRHs need to forward the quantized signals with sufficiently low error to the CP, which consumes more of their fronthaul capacities compared to a decoded message forwarded by a single RRH. On the other hand, when an RRH forwards a decoded message, the fronthaul capacity is significantly saved, but the receive-combining gain achievable with centralized decoding at the CP is compromised, and only a selection diversity gain over the RRHs can be attained. Besides the selection of DaF or FaD processing modes, the achievable rate for each SC also depends on the user-SC assignment and transmit power allocations by the users. Thus, the RRHs' processing mode~(DaF/FaD/no processing) selections on each SC should be jointly optimized along with the resource allocation in the network. 
The main results of this paper are summarized as follows.
\begin{itemize} 
\item We study a joint resource allocation problem in UD-CRAN, including the RRHs' processing mode selections, user-SC assignments and the users' transmit power allocations for OFDMA transmission, to maximize the users' weighted-sum-rate in the uplink transmission. To the best of our knowledge, this problem has not been considered yet in the literature. The problem, however, is non-convex, and an exhaustive search over all possible solutions would incur an exponential complexity of $O((K(2^M+M))^N)$, which is evidently not practically affordable in a UD-CRAN with large values of $N$ and/or $M$.
\item Hence, we propose a Lagrange duality based algorithm, which can achieve the optimal solution asymptotically when the number of SCs is large, at a much reduced complexity of $O(NK(2^M+M))$. 
\item To further reduce the complexity in UD networks, we propose a greedy algorithm to obtain a suboptimal solution, which has a lower complexity of only $O(NK(M^2+M))$, and is shown to be able to achieve close-to-optimal throughput performance under various practical setups by simulations. 
\item Finally, we compare the proposed optimal and suboptimal algorithms with hybrid decoding to both the conventional CRAN that performs FaD processing on all SCs, and a cellular network that performs DaF processing on all SCs by simulations, which show promising throughput gains by the proposed algorithms.
\end{itemize}

{\it Notation}: In this paper, scalars are denoted by lower-case letters, e.g., $x$, while vectors are denoted by bold-face lower-case letters, e.g., $\bm x$. The set of real numbers, complex numbers, and integers are denoted by $\mathbbm{R}$, $\mathbbm{C}$ and $\mathbbm{Z}$, respectively. Similarly, $(\cdot)^{x\times 1}$ denotes the corresponding space of $x$-dimensional column vectors, while $\mathbbm{R}_+$ and $\mathbbm{Z}_{++}$ denote the sets of non-negative real numbers and positive integers respectively. For a real scalar $x\in\mathbbm{R}$, $[x]^+\triangleq\max\{x,0\}$, 
and for a complex scalar $x\in\mathbbm{C}$, $|x|\geq 0$ denotes the magnitude of $x$. 
For a vector $\bm x$, $\bm x^\mathsf{T}$ denotes its transpose. Vectors with all elements equal to $1$ and $0$ are denoted by $\bm 1$ and $\bm 0$, respectively. For 
real vectors $\bm x,\bm y\in\mathbbm{R}^{M\times 1}$, 
$\bm x\succeq\bm y$ 
denotes the 
component-wise inequalities $x_i\geq y_i, i=1,\dotsc,M$. Finally, $\mathcal{CN}(\mu,\sigma^2)$ denotes 
a circularly symmetric complex Gaussian~(CSCG) random variable with mean $\mu$ and variance $\sigma^2$, and the symbol $\sim$ is used to mean  ``distributed as". 
\section{System Model}\label{Sec:SysMod}
\begin{figure}[t]
\centering
\includegraphics[width=\linewidth]{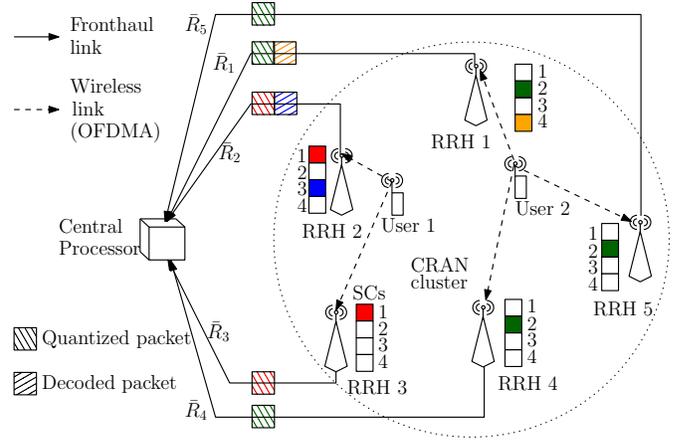}
\caption{OFDMA-based uplink CRAN with hybrid-decoding RRHs.}\label{F:SysModel}
\end{figure}
We study the uplink transmission in a UD-CRAN cluster that consists of $M$ single-antenna RRHs, denoted by the set $\mathcal{M}\triangleq\{1,\dotsc,M\}$, and $K$ single-antenna users, denoted by $\mathcal{K}\triangleq\{1,\dotsc,K\}$, as shown in~\cref{F:SysModel}. The users transmit their signals in the uplink using OFDMA over a total bandwidth of $B$~Hz 
which is equally divided into $N$ orthogonal SCs, denoted by the set 
$\mathcal{N}\triangleq\{1,\dotsc,N\}$. 
These signals are received by the RRHs, and are then forwarded 
to the CP via the individual fronthaul links of the RRHs. As the fronthaul link capacity is practically limited for each RRH, one of the following operations are chosen to be performed on each SC $n\in\mathcal{N}$ by each RRH $m\in\mathcal{M}$: 
\begin{itemize}
\item The user's message is decoded and forwarded to the CP~(DaF); 
\item The user's signal is quantized and forwarded to the CP, which then decodes the message~(FaD); or 
\item The received signal is not processed. 
\end{itemize}
For example, an illustration is given in~\cref{F:SysModel} with 5~RRHs and 2~users, where the SCs~1 and~3 are assigned to user~1, while SCs~2 and~4 are assigned to user~2. In this case, RRH~2 which is nearest to user~1, decodes the message in SC~3 and forwards it to the CP. Since this message is decoded by RRH~2, the other RRHs do not process the signal on SC~3. On the other hand, user~1's signal in SC~1 is quantized independently by each of the RRHs~2~and~4, and forwarded to the CP, which then decodes the message in SC~1 by combining the quantized signals from these two RRHs. Similarly for user~2, while RRH~1 locally decodes and forwards the message in SC~4, RRHs~1, 4~and~5 quantize and forward their respective received signals in SC~2. Notice that RRH~1, despite its proximity to user~1, does not process user~1's signal on SC~1, since its fronthaul is already fully consumed by user~2's quantized signal in SC~2 and the decoded message in SC~4. 

In order to describe the choice of operations outlined above completely, we introduce the following decision variables. First, to indicate whether an RRH $m$ processes the signal on SC $n$ or not, we define the variables $\alpha_{m,n}$ 
as  
\begin{align}
\alpha_{m,n}\triangleq\begin{cases}1&\text{if RRH}~m~\text{processes the signal on SC}~n\\
0&\text{otherwise,}\end{cases}
\label{E:ProcSel}
\end{align}
$m\in\mathcal{M}$ and $n\in\mathcal{N}$. The collection of these variables for all RRHs at each SC $n\in\mathcal{N}$ is denoted by the vector $\bm\alpha_n\triangleq\begin{bsmallmatrix}\alpha_{1,n}&\cdots&\alpha_{M,n}\end{bsmallmatrix}^\mathsf{T}\in\{0,1\}^{M\times 1}$. 
Next, to indicate the mode of processing chosen for SC $n$, i.e. whether the message in SC $n$ is decoded by an RRH or instead the signal in SC $n$ is quantized and forwarded to the CP, we define the mode selection 
variable $\delta_n$ as 
\begin{align}
\delta_n\triangleq\begin{cases}1&\text{if the message in SC}~n~\text{is decoded 
}\\
0&\text{if the signal in SC}~n~\text{is quantized, 
}\end{cases}\label{E:ModeSel}
\end{align}
$n\in\mathcal{N}$. 
In the DaF mode~$(\delta_n=1)$, at most one RRH in the network processes the 
signal in SC $n$, and hence we impose the condition $\bm 1^\mathsf{T}\bm\alpha_n\leq 1$. 
On the other hand, 
in the FaD mode~
$(\delta_n=0)$, a subset $\mathcal{A}_n\subseteq\mathcal{M}$ consisting of one or more RRHs, independently quantize and forward the signals in SC $n$, where $\mathcal{A}_n\triangleq\{m\in\mathcal{M}|\alpha_{m,n}=1\},~n\in\mathcal{N}$.
Finally, we define the variable $\nu_{k,n}$ to indicate whether an SC $n$ is assigned to a user $k$ or not, i.e., 
\begin{align}
\nu_{k,n}\triangleq\begin{cases}1&\text{if SC}~n~\text{is assigned to user}~k\\
0&\text{otherwise},
\end{cases}\label{E:UserSel}
\end{align}
$k\in\mathcal{K},n\in\mathcal{N}$. The vector $\bm\nu_n\triangleq\begin{bsmallmatrix}
\nu_{1,n}&\cdots&\nu_{K,n}
\end{bsmallmatrix}^\mathsf{T}\in\{0,1\}^K$ specifies the user assignment for SC $n$. According to OFDMA, each SC $n\in\mathcal{N}$ is assigned to at most one user for uplink transmission, and thus 
$\bm\nu_n$ must satisfy the condition $\bm 1^\mathsf{T}\bm\nu_n\leq 1,\forall n\in\mathcal{N}$. Also, we denote the set of SCs assigned to user $k\in\mathcal{K}$ by $\mathcal{N}_k\triangleq\{n|\nu_{k,n}=1\}\subseteq\mathcal{N}$, so that 
$\mathcal{N}_j\cap\mathcal{N}_k=\emptyset,\ \forall j\neq k,\ j,k\in\mathcal{K}$.

Let $h_{m,k,n}\in\mathbbm{C}$ denote the complex wireless channel coefficient from the user $k\in\mathcal{K}$ to RRH $m\in\mathcal{M}$, for SC $n\in\mathcal{N}$. We assume that all the channel coefficients $h_{m,k,n}$'s are known at the CP. Let $p_{k,n}\geq 0$ denote the uplink transmit power allocated by the user $k\in\mathcal{K}$ to SC $n\in\mathcal{N}$ and the vector $\bm p_n\triangleq\begin{bsmallmatrix}
p_{1,n}&\cdots&p_{K,n}
\end{bsmallmatrix}^\mathsf{T}\in\mathbbm{R}_+^K$ denote the transmit powers 
of all the users on SC $n\in\mathcal{N}$. 
Then, the 
signal received at RRH $m$ in SC $n\in\mathcal{N}_k$ 
is 
given by
\begin{align}
y_{m,n}=h_{m,k,n}\sqrt{p_{k,n}}s_{k,n}+z_m,
\label{E:RxSig}
\end{align} where $s_{k,n}\sim\mathcal{CN}(0,1)$ denotes user $k$'s information-bearing signal that is assumed to be complex Gaussian 
and $z_m\sim\mathcal{CN}(0,\sigma_m^2)$ is the additive-white-Gaussian-noise~(AWGN), with $\sigma_m^2$ denoting the receiver noise power at RRH $m$~(assumed to be equal for all SCs). 
\subsection{Forward-and-Decode~(FaD) Processing}\label{SS:JD}
When FaD processing~$(\delta_n=0)$ is performed on an SC $n\in\mathcal{N}_k$, the following operations take place. First, the RRHs $m\in\mathcal{A}_n$, i.e. with $\alpha_{m,n}=1$, quantize the received baseband signals $y_{m,n}$'s in SC $n$ given by~\eqref{E:RxSig}. Then, the RRHs encode these quantized values into digital codewords, and forward them 
to the CP. The CP recovers the quantized signals from these digital codewords, and decodes user $k$'s message on SC $n\in\mathcal{N}_k$, by combining the quantized signals from the RRHs. As the received signals at the RRHs in all the SCs are independent of each other and each RRH is assumed to perform signal quantization independently, a simple scalar quantization~(SQ) on the received signals $y_{m,n}$'s 
is performed~\cite{liu-etal2015joint}, and the baseband signal after SQ can be expressed as
\begin{align}
\hat{y}_{m,n}=y_{m,n}+e_{m,n},\ m\in\mathcal{M},n\in\mathcal{N},\label{E:QuantRxSig}
\end{align}
where $e_{m,n}$ denotes the complex-valued error induced by the SQ, which is assumed to have zero mean and a variance denoted by $q_{m,n}$. The errors $e_{m,n}$'s are independent over $m$ and $n$ due to the independent SQ in each SC by each RRH. A practical method for SQ is to represent the in-phase~(I) and quadrature~(Q) components of the received complex baseband signal $y_{m,n}$ in~\eqref{E:RxSig} using $\beta_m\in\mathbbm{Z}_{++}$ bits for each.\footnote{It is assumed that the resolution of the SQ, $\beta_m$, is fixed, and is the same on all SCs for a given RRH $m\in\mathcal{M}$.} In other words, each of the I~and~Q components of $y_{m,n}$ are uniformly quantized into one of $2^{\beta_m}$ levels~\cite{liu-etal2015joint}. With such a uniform SQ, 
the variance $q_{m,n}$ of the quantization error $e_{m,n}$ in~\eqref{E:QuantRxSig} on SC $n\in\mathcal{N}_k$ at RRH $m\in\mathcal{M}$ can be approximated as~\cite{liu-etal2015joint} 
\begin{align}
q_{m,n}=3(|h_{m,k,n}|^2p_{k,n}+\sigma_m^2)2^{-2\beta_m}
.\label{E:QNVar}
\end{align}

Next, each RRH $m\in\mathcal{A}_n$ transmits the digital code-words corresponding to the quantization levels to the CP via its fronthaul link. Since each I and Q component of the signal in an SC is represented by a $\beta_m$-bit codeword, it is not difficult to show that the transmission rate in bits-per-second~(bps) required on the fronthaul link of RRH $m$ to forward the quantized signal in any SC $n$ is given by $2B\beta_m/N$. Upon receiving the digital code-words, the CP first recovers the baseband quantized signals $\hat{y}_{m,n}$'s based on the quantization code-books used by each RRH. Further, to decode the message in SC $n$, the CP applies a linear combining to $\hat{y}_{m,n}$'s. When the optimal combining weights that maximize the signal-to-noise ratio~(SNR) are used, the received SNR for decoding at the CP on SC $n\in\mathcal{N}_k$ can be expressed as~\cite{liu-etal2015joint} 
\begin{align}
\gamma^Q_{k,n}(\bm\alpha_n,p_n)&=\sum_{m\in\mathcal{M}}\frac{\alpha_{m,n}|h_{m,k,n}|^2p_n}{\sigma_m^2+q_{m,n}}\notag\\
&=\sum_{m\in\mathcal{M}}\alpha_{m,n}\gamma^Q_{m,k,n}(p_{k,n}),
\label{E:RxSNRJD}
\end{align}
where $q_{m,n}$ is 
given by~\eqref{E:QNVar} and $\gamma^Q_{m,k,n}(p_{k,n})$ denotes the contribution of each RRH $m$ to the SNR at the CP in SC $n$. Note that if $\alpha_{m,n}=0$ for some RRH $m\in\mathcal{M}$, it means RRH $m$ does not process the signal on SC $n$ and hence does not contribute to the received SNR at the CP.
Using~\eqref{E:QNVar} in~\eqref{E:RxSNRJD}, $\gamma^Q_{m,k,n}(p_{k,n})$ can be expressed as 
\begin{align}
\gamma^Q_{m,k,n}(p_{k,n})=\frac{|h_{m,k,n}|^2p_{k,n}}{\sigma_m^2+3\left(|h_{m,k,n}|^2p_{k,n}+\sigma_m^2\right)2^{-2\beta_m}},
\label{E:PartSNR}
\end{align}
$m\in\mathcal{M},k\in\mathcal{K},n\in\mathcal{N}_k$. If the quantization error $e_{m,n}$ is assumed to be the worst-case Gaussian distributed, 
then a lower bound on the achievable rate in bps with FaD processing for SC $n\in\mathcal{N}_k$ 
is given by 
\begin{align}
&r^Q_{k,n}(\bm\alpha_n,p_{k,n})=(B/N)\log_2(1+\gamma^Q_{k,n}(\bm\alpha_n,p_{k,n}))
.\label{E:SCRJ}
\end{align}
Next, we present the following lemma.
\begin{lemma}\label{L:rknQConc}
With given RRH selection $\bm\alpha_n$, 
$r^Q_{k,n}(\bm\alpha_n,p_{k,n})$ defined in~\eqref{E:SCRJ} is concave in $p_{k,n}\geq 0$.
\end{lemma}
\begin{proof}
Please refer to~appendix~\ref{App:ProofLrknQConc}.
\end{proof}
\cref{L:rknQConc} indicates that on each SC $n$, the achievable rate with FaD processing is a concave function of the assigned user's transmit power. 
\subsection{Decode-and-Forward~(DaF) Processing}
Instead of having the CP perform the decoding by combining the quantized signals from the RRHs, any single selected RRH $m$ can also decode the message in SC $n$ locally, and then forward this message to the CP. The maximum achievable rate in bps for SC $n\in\mathcal{N}_k$ when RRH $m$ locally decodes user $k$'s message is given by
\begin{align}
r^D_{m,k,n}(p_{k,n})=(B/N)\log_2(1+|h_{m,k,n}|^2p_{k,n}/\sigma_m^2)
.\label{E:SCRL}
\end{align}
In this case, RRH $m$ forwards the decoded message on SC $n\in\mathcal{N}_k$ to the CP over its fronthaul link at the rate of at least $r^D_{m,k,n}(p_{k,n})$. Also, notice that if some RRH $m\in\mathcal{M}$ locally decodes the message in SC $n$ and forwards it to the CP, no other RRHs need to process the signal on this SC. 
\subsection{Hybrid Decoding}
Here, we present the proposed hybrid decoding scheme, where either a subset of RRHs $\mathcal{A}_n\subseteq\mathcal{M}$ quantize and forward their received signals to be 
jointly decoded at the CP as in a conventional CRAN, or a single RRH $m\in\mathcal{M}$ locally decodes a user's message on SC $n$ and forwards it to the CP. 
Combining the expressions for the achievable rates in these two cases as given in~\cref{E:SCRJ,E:SCRL} with the indicator variables defined in~\cref{E:ModeSel,E:ProcSel}, the achievable rate with hybrid decoding on an SC $n\in\mathcal{N}_k,\,k\in\mathcal{K},$ 
is given by 
\begin{align}
r_{k,n}(\delta_n,\bm\alpha_n,p_{k,n})&=\delta_n\sum_{m\in\mathcal{M}}\alpha_{m,n}r^D_{m,k,n}(p_{k,n})\notag\\
&\quad+(1-\delta_n)r^Q_{k,n}(\bm\alpha_n,p_{k,n}).
\label{E:SCR}
\end{align}
Note that if $\bm\alpha_n=\bm 0$, no RRH processes the received signals and the achievable rate on SC $n$ 
will be zero, irrespective of the value of 
$\delta_n$. Also, for an SC $n\in\mathcal{N}_k$, if $p_{k,n}=0$, we assume $\bm\alpha_n=\bm 0$ without loss of generality. 

Let $\bar{R}_m$ denote the 
fronthaul capacity of the link connecting RRH $m$ to the CP in bps. Then, the total rate at which RRH $m$ 
forwards the processed signals~(decoded and/or quantized) over all the $N$ SCs from the users to the CP 
must satisfy the constraints 
\begin{align}
&\sum_{n\in\mathcal{N}}\bigg(\delta_n\alpha_{m,n}\sum_{k\in\mathcal{K}}\nu_{k,n}r^D_{m,k,n}(p_{k,n})\notag\\
&\qquad+(1-\delta_n)\frac{2B\beta_m\alpha_{m,n}}{N}\bigg)\leq \bar{R}_m,\,\forall m\in\mathcal{M}.
\label{E:IFHRC}
\end{align}
The left-hand side of~\eqref{E:IFHRC} represents the total rate at which the signals processed by RRH $m$ must be  forwarded to the CP, 
where $r^D_{m,k,n}(p_{k,n})$ is defined in~\eqref{E:SCRL}, and $2B\beta_m/N$ is the rate in bps corresponding to the resolution $\beta_m$ of the SQ performed by RRH $m$. 
In the next section, we formulate the joint uplink resource allocation optimization problem for the OFDMA-based UD-CRAN with hybrid decoding.
\section{Problem Formulation}\label{Sec:ProbForm}
We aim to maximize the weighted-sum-rate of the users in the uplink over all the SCs, by jointly optimizing 
the processing mode and RRH selections, which are collectively represented by the variables $\{\delta_n,\bm\alpha_n\}_{n\in\mathcal{N}}$, the user-SC assignments $\{\bm\nu_n\}_{n\in\mathcal{N}}$, and the users' transmit power allocations $\{\bm p_n\}_{n\in\mathcal{N}}$, subject to the fronthaul constraints~\eqref{E:IFHRC}, and the total power constraints at each of the users, which we denote by $\bar{P}_k,\ k\in\mathcal{K}$. Let $\omega_k\geq 0$ denote the rate weight for user $k\in\mathcal{K}$. Then the  problem can be formulated as given below. 
\begin{subequations}
\label{P:ULHybMain}
\begin{align}
\mathop{\mathrm{maximize}}_{\{\bm p_n,\bm\alpha_n,\bm\nu_n,\delta_n\}_{n\in\mathcal{N}}
}&\sum_{k\in\mathcal{K}}\omega_k\sum_{n\in\mathcal{N}}\nu_{k,n}r_{k,n}(\delta_n,\bm\alpha_n,p_{k,n})\tag{\ref{P:ULHybMain}}\\
\mathrm{subject}~\mathrm{to}\notag\\
&\ \text{\cref{E:IFHRC}}\notag\\
&\ \sum_{n\in\mathcal{N}} p_{k,n}\leq\bar{P}_k\quad\forall k\in\mathcal{K}
\label{C:PCULHybMain}\\
&\ p_{k,n}\geq 0\quad\forall k\in\mathcal{K},\forall n\in\mathcal{N}\label{C:PosPowULHybMain}\\
&\ \delta_n(\bm 1^\mathsf{T}\bm \alpha_n)\leq 1\quad\forall n\in\mathcal{N}\label{C:LDConstULHybMain}\\
&\ \alpha_{m,n}\in\{0,1\}\quad\forall m\in\mathcal{M},\forall n\in\mathcal{N}\label{C:RRHSelULHybMain}\\
&\ \bm 1^\mathsf{T}\bm\nu_n\leq 1\quad\forall n\in\mathcal{N}\label{C:UASUULHybMain}\\
&\ \nu_{k,n}\in\{0,1\}\quad\forall k\in\mathcal{K},\forall n\in\mathcal{N}\label{C:UAVectULHybMain}\\
&\ \delta_n\in\{0,1\}\quad\forall n\in\mathcal{N}.\label{C:ModeSelULHybMain}
\end{align}
\end{subequations}
The constraint~\eqref{C:LDConstULHybMain} ensures that when DaF processing is performed~$\left(\delta_n=1\right)$, at most one RRH processes the signal on SC $n$, i.e. at most one $\alpha_{m,n}=1$, 
and all other $\alpha_{m,n}$'s are zero. In the case of FaD processing~
$\left(\delta_n=0\right)$, no such restriction is imposed on $\bm\alpha_n$. The above weighted-sum-rate maximization problem is non-convex due to the integer constraints~\cref{C:RRHSelULHybMain,C:ModeSelULHybMain} and the coupled variables in the objective, as well as in constraints~\cref{E:IFHRC,C:LDConstULHybMain}. Notice that if $\delta_n=0\ \forall n\in\mathcal{N}$, then problem~\cref{P:ULHybMain} reduces to the special case where all RRHs quantize and forward their signals to the CP as in an OFDMA-based CRAN with only centralized decoding at the CP~\cite{liu-etal2015joint}. On the other hand, if $\delta_n=1\ \forall n\in\mathcal{N}$, 
problem~\cref{P:ULHybMain} reduces to the special case of an OFDMA-based cellular network where the user messages on different SCs can be decoded by different RRHs in general and then forwarded to the CP over the fronthauls. This is unlike the conventional cellular network where each user's signal is decoded by only the single BS to which it is associated. In the general case, the optimal solution to the above problem could have $\delta_n=0$ only on a subset of SCs, and $\delta_n=1$ for other SCs, in order to maximally exploit the hybrid decoding capability. 

It is worth noting that in the solution to problem~\eqref{P:ULHybMain}, if $p_{k,n}=0,~\forall n\in\mathcal{N}$ for some user $k$, it means that user $k$ is not served at all in the current scheduling interval. The resulting fairness issue can be overcome in practice by selecting the user rate-weights $\omega_k$'s in problem~\eqref{P:ULHybMain} appropriately. For example, if one user receives very low rate in the current scheduling interval, its priority weight can be increased to ensure that it gets a higher rate in the next interval. The user weights can thus be manipulated to ensure that all the users in the network are fairly served in the long term. 
\section{Proposed Solutions}\label{Sec:PropSol}
\subsection{Optimal Solution}\label{SS:OptSol}
Although problem~\eqref{P:ULHybMain} is non-convex, 
its duality gap diminishes to zero as the number of SCs $N$ goes to infinity, as under this condition such non-convex problems are shown to satisfy the ``time-sharing" property~\cite{yu-lui2006dual}. Since $N$ is typically large in practice, we propose to apply the Lagrange duality method to obtain an asymptotically optimal solution to problem~\eqref{P:ULHybMain}.\footnote{However, from the simulations in~\cref{Sec:SimResults}, we observe that the duality gap is negligible even for a moderate value of $N=64$.} Let $\lambda_m\geq 0,\ m\in\mathcal{M}$ denote the dual variables associated with the $M$  constraints in~\eqref{E:IFHRC}, and $\mu_k\geq 0,\ k\in\mathcal{K}$, denote the dual variables for the $K$ constraints in~\eqref{C:PCULHybMain}. Also let the vectors $\bm \lambda\triangleq\begin{bsmallmatrix}\lambda_1&\cdots&\lambda_M\end{bsmallmatrix}^\mathsf{T}\in\mathbbm{R}_+^{M\times 1}$ and $\bm\mu\triangleq\begin{bsmallmatrix}
\mu_1&\cdots&\mu_K
\end{bsmallmatrix}^\mathsf{T}\in\mathbbm{R}_+^{K\times 1}$ denote the collections of these dual variables. 
Then, the~(partial) Lagrangian of problem~\eqref{P:ULHybMain} is given by 
\begin{align}
&L\left(\{\delta_n,\bm\nu_n,\bm\alpha_n,\bm p_n\}_{n\in\mathcal{N}},\bm\lambda,\bm\mu\right)\notag\\
&=\sum_{k\in\mathcal{K}}\omega_k\sum_{n\in\mathcal{N}}\nu_{k,n}r_{k,n}(\delta_n,\bm\alpha_n,p_{k,n})
-\sum_{m\in\mathcal{M}}\lambda_m\notag\\
&\quad\cdot\Bigg(\frac{1}{\bar{R}_m}\sum_{n\in\mathcal{N}}\Big(\delta_n\alpha_{m,n}\sum_{k\in\mathcal{K}}\nu_{k,n}r^D_{m,k,n}(p_{k,n})\notag\\
&\quad+(1-\delta_n)\frac{2B\beta_m\alpha_{m,n}}{N}\Big)-1\Bigg)
-\sum_{k\in\mathcal{K}}\mu_k\bigg(\sum_{n\in\mathcal{N}} p_{k,n}-\bar{P}_k\bigg)\notag\\
&=\sum_{n\in\mathcal{N}} L_n(\delta_n,\bm\nu_n,\bm\alpha_n,\bm p_n,\bm\lambda,\bm\mu)+\sum_{m\in\mathcal{M}}\lambda_m+\sum_{k\in\mathcal{K}}\mu_k\bar{P}_k,
\label{E:LagW}
\end{align} 
where each term in the first summation in~\eqref{E:LagW} is 
\begin{align}
&L_n\left(\delta_n,\bm\nu_n,\bm\alpha_n,\bm p_n,\bm\lambda,\bm\mu\right)\notag\\
&\triangleq\sum_{k\in\mathcal{K}}\omega_k\nu_{k,n}\bigg(\delta_n\sum_{m\in\mathcal{M}}\alpha_{m,n}r^D_{m,k,n}(p_{k,n})\notag\\
&\quad+(1-\delta_n)r^Q_{k,n}(\bm\alpha_n,p_{k,n})\bigg)
-\sum_{m\in\mathcal{M}}\frac{\lambda_m}{\bar{R}_m}\notag\\
&\quad\cdot\bigg(\delta_n\alpha_{m,n}\sum_{k\in\mathcal{K}}\nu_{k,n}r^D_{m,k,n}(p_{k,n})+(1-\delta_n)\frac{2B\beta_m\alpha_{m,n}}{N}\bigg)\notag\\
&\quad-\sum_{k\in\mathcal{K}}\mu_k p_{k,n},\quad n\in\mathcal{N}.\label{E:Lagn}
\end{align}
The Lagrange dual function can thus be expressed as
\begin{subequations}
\label{E:DualFunc} 
\begin{align}
g(\bm\lambda,\bm\mu)=\max_{\substack{\{\bm p_n,\bm\alpha_n,\\
\bm\nu_n,\delta_n\}_{n\in\mathcal{N}}}}&\,L\left(\{\delta_n,\bm\nu_n,\bm\alpha_n,\bm p_n,\}_{n\in\mathcal{N}},\bm\lambda,\bm\mu\right)\tag{\ref{E:DualFunc}}\\
\mathrm{s.t.}&\,\text{\cref{C:PosPowULHybMain,C:LDConstULHybMain,C:ModeSelULHybMain}}.\notag
\end{align}
\end{subequations}
The maximization problem in~\eqref{E:DualFunc} can be decomposed into $N$ parallel sub-problems using~\eqref{E:LagW}, where each sub-problem corresponds to a single SC $n$, and has the following structure,
\begin{subequations} 
\label{P:DualFuncn}
\begin{align}
\max_{\bm p_n,\bm\alpha_n,\bm\nu_n,\delta_n}&\:L_n(\delta_n,\bm\nu_n,\bm\alpha_n,\bm p_n,\bm\lambda,\bm\mu)\tag{\ref{P:DualFuncn}}\nonumber\\
\mathrm{s.t.}&\:\bm p_n\succeq\bm 0\label{C:PNZDFn}\\
&\ \delta_n(\bm 1^\mathsf{T}\bm\alpha_n)\leq 1\label{C:ModeRRHSelDFn}\\
&\ \bm\alpha_n\in\{0,1\}^{M\times 1}\label{C:RRHSelVectDFn}\\
&\ \bm 1^\mathsf{T}\bm\nu_n\leq 1\label{C:UASUDFn}\\
&\ \bm\nu_n\in\{0,1\}^{K\times 1}\label{C:UAVectDFn}\\
&\ \delta_n\in\{0,1\},\label{C:ModeDFn}
\end{align}
\end{subequations}
with the objective function given by~\eqref{E:Lagn}. Problem~\eqref{P:DualFuncn} is non-convex due to the integer constraints~\cref{C:RRHSelVectDFn,C:UAVectDFn,C:ModeDFn} 
and also due to the coupled variables in the objective and the constraint~\eqref{C:ModeRRHSelDFn}. However, since the processing mode selection variable $\delta_n$ takes on only two values 0 or 1, we solve problem~\eqref{P:DualFuncn} for each of these two cases separately, and then choose the 
value of $\delta_n$ that gives the maximum objective value for SC $n$. Specifically, let $\bar{L}^Q_n$ denote the optimal value of problem~\eqref{P:DualFuncn} for the FaD mode~$(\delta_n=0)$ and $\bar{L}^D_n$ denote the 
corresponding value for the DaF mode~$(\delta_n=1)$. Then, for given dual variables $\bm\lambda$ and $\bm\mu$, the optimal processing mode selection $\bar{\delta}_n$ for problem~\eqref{P:DualFuncn} is given by  
\begin{align}
\bar{\delta}_n=\begin{cases}
0&\text{if }\bar{L}^{Q}_n>\bar{L}^{D}_n\\
1&\text{otherwise.}\label{E:OptModeDFn}
\end{cases}
\end{align}
In the following two subsections, we explain in detail how to solve problem~\eqref{P:DualFuncn} for each processing mode, in order to obtain $\bar{L}^{Q}_n$ and $\bar{L}^{D}_n$.
\subsubsection{The Case of FaD Processing~$\left(\delta_n=0\right)$}\label{SS:FD}
With 
$\delta_n=0$, 
problem~\eqref{P:DualFuncn} reduces to 
\begin{subequations} 
\label{P:DFnJD}
\begin{align}
\max_{\bm p_n,\bm\alpha_n,\bm\nu_n}&\ L^Q\left(\bm\nu_n,\bm\alpha_n,\bm p_n,\bm\lambda,\bm\mu\right)\tag{\ref{P:DFnJD}}\\
\mathrm{s.t.}&\ \text{\cref{C:PNZDFn,C:RRHSelVectDFn,C:UASUDFn,C:UAVectDFn}},\notag
\end{align}
\end{subequations}
where the objective function is given by
\begin{align}
&L^Q(\bm\nu_n,\bm\alpha_n,\bm p_n,\bm\lambda,\bm\mu)\notag\\
&\triangleq\sum_{k\in\mathcal{K}}\nu_{k,n}\omega_kr_{k,n}^Q\left(\bm\alpha_n,p_{k,n}\right)
-\frac{2B}{N}\sum_{m\in\mathcal{M}}\frac{\beta_m\lambda_m\alpha_{m,n}}{\bar{R}_m}\notag\\
&\quad-\sum_{k\in\mathcal{K}}\mu_k p_{k,n}\label{E:ObjDFnJD}
\end{align}
Now, let the user assignment on SC $n$ be fixed as $\bm\nu_n=\hat{\bm\nu}_n$. 
If no user is assigned to SC $n$, i.e. $\hat{\bm\nu}_n=\bm 0$, it is evident from~\eqref{E:ObjDFnJD} that the objective of problem~\cref{P:DFnJD} 
is maximized by setting $\bm p_n=\bm 0$ and $\bm\alpha_n=\bm 0$, 
i.e. 
the power allocation is zero and no RRH processes the signal, as expected. Otherwise, if 
SC $n$ is assigned to a user $\hat{k}_n\in\mathcal{K}$, 
so that $\hat{\nu}_{\hat{k}_n,n}=1$, and $\hat{\nu}_{k,n}=0~\forall k\neq\hat{k}_n$, then problem~\eqref{P:DFnJD} can be written as 
\begin{subequations}
\label{P:DFnJDFixUA}
\begin{align}
\max_{\substack{p_{\hat{k}_n,n}\geq 0,\\
\bm\alpha_n\in\{0,1\}^M}}&\,\omega_{\hat{k}_n}r_{\hat{k}_n,n}^Q\big(\bm\alpha_n,p_{\hat{k}_n,n}\big)-\frac{2B}{N}\sum_{m\in\mathcal{M}}\frac{\beta_m\lambda_m\alpha_{m,n}}{\bar{R}_m}\notag\\
&\,-\mu_{\hat{k}_n}p_{\hat{k}_n,n}.\tag{\ref{P:DFnJDFixUA}}
\end{align}
\end{subequations}
Although the above problem~\eqref{P:DFnJDFixUA} is non-convex due to the integer variables $\bm\alpha_n$ and the coupled variables in the objective, when the selection of RRHs is also fixed as $\bm\alpha_n=\tilde{\bm\alpha}_n$, it reduces to the problem
\begin{subequations} 
\label{P:DFnJDFUARRHSel}
\begin{align}
\max_{p_{\hat{k}_n,n}\geq 0}&\:\omega_{\hat{k}_n}r_{\hat{k}_n,n}^Q\big(\tilde{\bm\alpha}_n,p_{\hat{k}_n,n}\big)-\mu_{\hat{k}_n} p_{\hat{k}_n,n},\tag{\ref{P:DFnJDFUARRHSel}}
\end{align}
\end{subequations}
which is convex due to~\cref{L:rknQConc}. In the special case when only one RRH quantizes and forwards the signal in SC $n$, i.e. when $\bm 1^\mathsf{T}\tilde{\bm\alpha}_n=1$, the optimal user power allocation $\tilde{p}_{\hat{k}_n,n}$ that solves problem~\eqref{P:DFnJDFUARRHSel} can be obtained in a closed-form as given by the following proposition.
\begin{proposition}\label{Prop:OptPASRRHQ}
Let 
$\tilde{m}_n\in\mathcal{M}$ be the single RRH that quantizes and forwards the signal in SC $n$, so that $\tilde{\alpha}_{\tilde{m}_n,n}=1$ and $\tilde{\alpha}_{m,n}=0\ \forall m\neq\tilde{m}_n$. Then, the optimal power allocation that solves problem~\eqref{P:DFnJDFUARRHSel} is given by
\begin{align}
\tilde{p}_{\hat{k}_n}&=\frac{\left(\theta_{\tilde{m}_n}+2\right)\sigma_{\tilde{m}_n}^2}{2|h_{\tilde{m}_n,\hat{k}_n,n}|^2}\Bigg(\bigg(1+\frac{4}{\left(\theta_{\tilde{m}_n}+2\right)^2}\notag\\
&\quad\cdot\bigg[\frac{\omega_{\hat{k}_n} B|h_{\tilde{m}_n,\hat{k}_n,n}|^2\theta_{\tilde{m}_n}}{\sigma_{\tilde{m}_n}^2\mu_{\hat{k}_n} N\ln 2}-\left(\theta_{\tilde{m}_n}+1\right)\bigg]^+\bigg)^{\frac{1}{2}}-1\Bigg),\label{E:OptPASingleRRH}
\end{align}
where $\theta_{\tilde{m}_n}\triangleq 2^{2\beta_{\tilde{m}_n}}/3>0$.
\end{proposition}
\begin{proof}
Please refer to appendix~\ref{App:ProofOptPASRRHQ}.
\end{proof}
\cref{Prop:OptPASRRHQ} shows that the optimal user power allocation in the case of FaD processing with a single RRH selection has a threshold structure and is non-zero only if $\tfrac{|h_{\tilde{m}_n,\hat{k}_n,n}|^2}{\sigma_{\tilde{m}_n}^2}>\tfrac{\mu_{\hat{k}_n} N\ln 2}{\omega_{\hat{k}_n}B}\big(1+\tfrac{1}{\theta_{\tilde{m}_n}}\big)$. 
On the other hand, if more than one RRHs quantize and forward the signal in SC $n$, i.e. if $\bm 1^\mathsf{T}\tilde{\bm\alpha}_n>1$, then 
the optimal user power allocation $\tilde{p}_{\hat{k}_n,n}$ for problem~\eqref{P:DFnJDFUARRHSel} can be efficiently obtained by a one-dimensional line search whose complexity does not depend on the number of RRHs $M$. Thus, for a fixed RRH selection $\tilde{\bm\alpha}_n$, $\tilde{p}_{\hat{k}_n,n}$ can be efficiently obtained either using~\eqref{E:OptPASingleRRH} when $\bm 1^\mathsf{T}\tilde{\bm\alpha}_n=1$, or via a one-dimensional line search, when $\bm 1^\mathsf{T}\tilde{\bm\alpha}_n>1$. 
The optimal solution to the joint RRH selection and user power allocation problem~\eqref{P:DFnJDFixUA} can thus be obtained by solving problem~\eqref{P:DFnJDFUARRHSel} 
for all the $2^M$ possible RRH selections $\tilde{\bm\alpha}_n\in\{0,1\}^M$, 
and then choosing the RRH selection with the corresponding optimal user power allocation that gives the largest objective value for~\eqref{P:DFnJDFixUA}. Since the optimal user power allocation can be obtained either in closed-form using~\eqref{E:OptPASingleRRH} when $\bm 1^\mathsf{T}\tilde{\bm\alpha}_n=1$, or using a one-dimensional line search whose complexity does not depend on $M$, when $\bm 1^\mathsf{T}\tilde{\bm\alpha}_n>1$, the overall complexity of optimally solving problem~\eqref{P:DFnJDFixUA} is $O\left(2^M\right)$. Finally, after solving problem~\eqref{P:DFnJDFixUA} for each user assignment $\hat{k}_n\in\mathcal{K}$, 
the optimal solution to problem~\eqref{P:DFnJD} for the FaD mode can be obtained by choosing the user assignment and corresponding RRH selection and power allocation that maximizes the objective in~\eqref{P:DFnJD}. Thus, 
the overall complexity of solving problem~\eqref{P:DFnJD} optimally is $O\left(K2^M\right)$. 
The next section describes how to solve problem~\eqref{P:DualFuncn} for the other DaF processing mode~$(\delta_n=1)$. 
\subsubsection{The Case of DaF Processing~$\left(\delta_n=1\right)$}\label{SS:DF}
In the DaF mode, $\delta_n=1$ and problem~\eqref{P:DualFuncn} reduces to 
\begin{subequations}
\label{P:DFnLD}
\begin{align}
\max_{\bm p_n\succeq\bm 0,\bm\alpha_n\in\{0,1\}^M}&\,L^D(\bm\nu_n,\bm\alpha_n,\bm p_n,\bm\lambda,\bm\mu)\tag{\ref{P:DFnLD}}\\
\mathrm{s.t.}
&\,\bm 1^\mathsf{T}\bm\alpha_n\leq 1,\label{C:RRHSelDFnLD}
\end{align}
\end{subequations}
where the objective function is given by
\begin{align}
&L^D(\bm\nu_n,\bm\alpha_n,\bm p_n,\bm\lambda,\bm\mu)\notag\\
&\triangleq\sum_{k\in\mathcal{K}}\omega_k\nu_{k,n}\sum_{m\in\mathcal{M}}\alpha_{m,n}r^D_{m,k,n}(p_{k,n})-\sum_{m\in\mathcal{M}}\frac{\lambda_m\alpha_{m,n}}{\bar{R}_m}\notag\\
&\quad\cdot\sum_{k\in\mathcal{K}}\nu_{k,n}r^D_{m,k,n}(p_{k,n})
-\sum_{k\in\mathcal{K}}\mu_k p_{k,n}.
\end{align}
Similar to the case of FaD processing in~\cref{SS:FD}, if no user is assigned to SC $n$, 
it is evident that the optimal power allocation for problem~\eqref{P:DFnLD} is $\bm p_n=\bm 0$, while it can be assumed without loss of generality that no RRH is selected for processing the signal, i.e. $\bm\alpha_n=\bm 0$, in this case. On the other hand, if SC $n$ is assigned to a user $\hat{k}_n\in\mathcal{K}$, then problem~\eqref{P:DFnLD} reduces to 
\begin{subequations}
\label{P:DFnLDFixUA}
\begin{align}
\max_{
p_{\hat{k}_n,n}\geq 0,
\bm\alpha_n\in\{0,1\}^M
}&\,\sum_{m\in\mathcal{M}}\alpha_{m,n}\Big(\omega_{\hat{k}_n}-\frac{\lambda_m}{\bar{R}_m}
\Big)r^D_{m,\hat{k}_n,n}\big(p_{\hat{k}_n,n}\big)\notag\\
&\,-\mu_{\hat{k}_n} p_{\hat{k}_n,n}\tag{\ref{P:DFnLDFixUA}}\\
\mathrm{s.t.}&\,\eqref{C:RRHSelDFnLD}.\notag
\end{align}
\end{subequations}
Problem~\eqref{P:DFnLDFixUA} is non-convex due to the integer constraints on the RRH selection $\bm\alpha_n$ and hence difficult to solve directly to find the optimal power allocation for user $\hat{k}_n$ and the optimal RRH that should locally decode this user's signal. Thus, similar to the procedure in~\cref{SS:FD}, we first assume that an RRH $\tilde{m}_n\in\mathcal{M}$ is chosen to decode user $\hat{k}_n$'s message in SC $n$ i.e., $\alpha_{\tilde{m}_n,n}=1$ and $\alpha_{m,n}=0\ \forall m\neq\tilde{m}_n$, which reduces problem~\eqref{P:DFnLDFixUA} to
\begin{subequations}
\label{P:DFnLDFUARRHSel}
\begin{align}
\max_{p_{\hat{k}_n,n}\geq 0}&\,\Big(\omega_{\hat{k}_n}-
\frac{\lambda_{\tilde{m}_n}}{\bar{R}_{\tilde{m}_n}}
\Big)r^D_{\tilde{m}_n,n}\big(p_{\hat{k}_n,n}\big)-\mu_{\hat{k}_n} p_{\hat{k}_n,n}\tag{\ref{P:DFnLDFUARRHSel}}.
\end{align}
\end{subequations}
The optimal solution to the above problem~\eqref{P:DFnLDFUARRHSel} is given by the following proposition. 
\begin{proposition}\label{Prop:OptPDFnLDFUARRHSel}
The optimal power allocation $\tilde{p}_{\hat{k}_n,n}$ that solves problem~\eqref{P:DFnLDFUARRHSel} is given by 
\begin{align}
\tilde{p}_{\hat{k}_n,n}=\Bigg[\frac{B}{\mu N\ln 2}\Big(\omega_{\hat{k}_n}-\frac{\lambda_{\tilde{m}_n}}{\bar{R}_{\tilde{m}_n}}\Big)-\frac{\sigma_{\tilde{m}_n}^2}{|h_{\tilde{m}_n,\hat{k}_n,n}|^2}\Bigg]^+.\label{E:OptPALD}
\end{align} 
\end{proposition}
\begin{proof}
The proof is similar to that for the standard water-filling power allocation for parallel AWGN channels~\cite{goldsmith2005wireless}.  
\end{proof}
The optimal power allocation in~\eqref{E:OptPALD} for a fixed user assignment and RRH selection has the same form as the 
classic water-filling solution~\cite{goldsmith2005wireless}. However, in general 
the water levels are different for different SCs, and are determined by the user $\hat{k}_n$ to which SC $n$ is assigned to, as well as the RRH $\tilde{m}_n$ that is chosen to decode user $\hat{k}_n$'s message, through the parameters $\omega_{\hat{k}_n}$, $\lambda_{\tilde{m}_n}$, $\bar{R}_{\tilde{m}_n}$ and $h_{\tilde{m}_n,\hat{k}_n,n}$. The water-level on SC $n$ can be negative if $\lambda_{\tilde{m}_n}/\bar{R}_{\tilde{m}_n}\geq \omega_{\hat{k}_n}$, which implies that no power should be allocated to user $\hat{k}_n$ on SC $n$ in this case.
The optimal solution to the joint RRH selection and user power allocation problem~\eqref{P:DFnLDFixUA} can thus be obtained by computing $\tilde{p}_{\hat{k}_n,n}$ as given in~\eqref{E:OptPALD} for all the RRHs $\tilde{m}_n\in\mathcal{M}$, and then selecting the RRH that gives the highest objective value for problem~\eqref{P:DFnLDFixUA}. This procedure incurs a complexity of $O\left(M\right)$ and is performed for each user assignment $\hat{k}_n\in\mathcal{K}$. Then the optimal solution to problem~\eqref{P:DFnLD} for the DaF mode can be found by choosing the user that gives the highest objective value in~\eqref{P:DFnLD}. 
The overall complexity of optimally solving problem~\eqref{P:DFnLD} on each SC $n$ for the DaF processing mode is thus $O(KM)$. 

Finally, after solving the problems~\cref{P:DFnJD,P:DFnLD} for the FaD and DaF processing modes to obtain the maximal objective values $\bar{L}^Q_n$ and $\bar{L}^D_n$, respectively, the optimal solution to problem~\eqref{P:DualFuncn} can be found by selecting the processing mode 
that maximizes the objective in~\eqref{E:Lagn}, according to~\eqref{E:OptModeDFn}. 
Thus, for given dual variables $\bm\lambda$ and $\bm\mu$, problem~\eqref{P:DualFuncn} can be solved for each SC $n\in\mathcal{N}$ optimally, incurring a worst-case complexity of $O(K(2^M+M))$. 
Next, we consider the dual problem corresponding to~\eqref{P:ULHybMain}, given by 
\begin{align}
\min_{\bm\lambda\succeq\bm 0,\bm\mu\succeq\bm 0}g(\bm\lambda,\bm\mu).\label{E:DualProb}
\end{align}
The above dual problem~\eqref{E:DualProb} is convex and can be solved efficiently by using e.g., the ellipsoid method~\cite{boyd2014ellipsoid} to find the optimal dual variables $\bm\lambda^\star$ and $\bm\mu^\star$. Then, the optimal solution to problem~\eqref{P:DualFuncn} on each SC $n\in\mathcal{N}$ is given by $(\delta^\star_n,\bm\nu^\star_n,\bm\alpha_n^\star,\bm p_n^\star)$, computed as explained above with the optimal dual variables $\bm\lambda^\star$ and $\bm\mu^\star$. The overall algorithm for solving the joint uplink resource allocation problem~\eqref{P:ULHybMain} is thus summarized in~\cref{A:Overall}. Since the complexity of the ellipsoid method to find the optimal dual variables depends only on the size of the initial ellipsoid and the maximum length of the sub-gradients over the intial ellipsoid~\cite{boyd2014ellipsoid}, 
it follows from the above discussion that an asymptotically optimal solution to problem~\eqref{P:ULHybMain} for large $N$ can be efficiently computed with an overall complexity of $O(NK(2^M+M))$ using the algorithm in~\cref{A:Overall}. 
\begin{table}[h]
\centering
\caption{Algorithm for problem~\eqref{P:ULHybMain}}\label{A:Overall}
\begin{framed}
\begin{algorithmic}[1]
\State Initialization: $\bm\lambda\succeq\bm 0$, $\bm\mu\succeq\bm 0$
\Repeat 
\For {each SC $n\in\mathcal{N}$}
\State With $\delta_n=0$, for each user $\hat{k}_n\in\mathcal{K}$ and each RRH selection $\tilde{\bm\alpha}_n\in\{0,1\}^M$, find optimal user power allocation by solving~\eqref{P:DFnJDFUARRHSel}\label{AL:JDOptRRHPA}
\State Choose RRH selection and corresponding user power allocation that maximizes objective in~\eqref{P:DFnJDFixUA} 
\State Choose user with optimal RRH selection and power allocation that gives maximum value $\bar{L}^Q_n$ of objective in~\eqref{P:DFnJD}
\State With $\delta_n=1$, for each user $\hat{k}_n\in\mathcal{K}$ and each RRH $\tilde{m}_n\in\mathcal{M}$, find optimal user power allocation using~\cref{E:OptPALD} 
\State Choose RRH and corresponding user power allocation that maximizes objective in~\eqref{P:DFnLDFixUA}
\State Choose user with optimal RRH and power allocation that gives maximum value $\bar{L}^D_n$ of objective in~\eqref{P:DFnLD}
\State Find optimal mode selection $\bar{\delta}_n$ using~\eqref{E:OptModeDFn}
\EndFor
\State Update dual variables $\bm\lambda$, $\bm\mu$ using the ellipsoid method
\Until ellipsoid algorithm converges to desired accuracy
\end{algorithmic}
\end{framed}
\end{table}
\subsection{Suboptimal Solution}\label{SS:SubOptSubP}
In~\cref{SS:OptSol}, it is observed that optimally solving the problem~\eqref{P:DFnJDFixUA} for the FaD mode on each SC $n$ for a fixed user assignment requires an exhaustive search over all possible RRH selections, incurring a complexity of $O(2^M)$, which may be unsuitable for dense networks with large values of $M$. In order to reduce this complexity, we propose an alternative way to solve problem~\eqref{P:DFnJDFixUA} suboptimally, via a greedy algorithm, in this subsection. 
The objective function of problem~\eqref{P:DFnJDFixUA} can be written in terms of the set of selected RRHs $\mathcal{A}_n$, as
\begin{align}
f\big(\mathcal{A}_n,p_{\hat{k}_n,n}\big)&\triangleq\omega_kr_{\hat{k}_n,n}^Q\big(\mathcal{A}_n,p_{\hat{k}_n,n}\big)-\frac{2B}{N}\sum_{m\in\mathcal{A}_n}\frac{\beta_m\lambda_m}{\bar{R}_m}\notag\\
&\quad-\mu_{\hat{k}_n}p_{\hat{k}_n,n}.\label{E:ObjFixUASet}
\end{align}
Let $p_{\hat{k}_n,n}(\mathcal{A}_n)$ denote the optimal power allocation for user $\hat{k}_n$ that solves problem~\eqref{P:DFnJDFUARRHSel} when the selected set of RRHs is $\mathcal{A}_n$, which can be obtained using~\eqref{E:OptPASingleRRH} when $|\mathcal{A}_n|=1$, or by a line-search otherwise. Then, we construct a suboptimal set $\check{\mathcal{A}}_n$ and find the corresponding user power allocation $\check{p}_{\hat{k}_n,n}\triangleq p_{\hat{k}_n,n}\big(\check{\mathcal{A}_n}\big)$ for problem~\eqref{P:DFnJDFixUA} using a greedy algorithm as follows. Let $\mathcal{A}_{n,i}$ denote the set of selected RRHs on SC $n$ at the start of an iteration $i$, and $f_i\triangleq f\big(\mathcal{A}_{n,i},p_{\hat{k}_n,n}(\mathcal{A}_{n,i})\big)$ denote the corresponding maximum objective value of problem~\eqref{P:DFnJDFixUA} as given by~\eqref{E:ObjFixUASet}. 
We assume that initially, $\mathcal{A}_{n,1}=\emptyset$, and 
$f_1=0$. 
At each iteration $i=1,\dotsc,M$, we first find an RRH $j_i\in\mathcal{M}\setminus\mathcal{A}_{n,i}$, which is not currently selected and when added to the current set of RRHs $\mathcal{A}_{n,i}$, maximizes the objective in~\eqref{E:ObjFixUASet} among all the currently un-selected RRHs, i.e., 
\begin{align}
j_i=\mathop{\arg\max}_{\ell\in\mathcal{M}\setminus\mathcal{A}_{n,i}}f\big(\mathcal{A}_{n,i}\cup\{\ell\},p_{\hat{k}_n,n}(\mathcal{A}_{n,i}\cup\{\ell\})\big).\label{E:MaxRRHji}
\end{align}
If RRH $j_i$ improves the current maximum objective value $f_{i-1}$, we add it to the current set of selected RRHs $\mathcal{A}_{n,i}$, i.e., if 
\begin{align}
f\big(\mathcal{A}_{n,i}\cup\{j_i\},p_{\hat{k}_n,n}(\mathcal{A}_{n,i}\cup\{j_i\})\big)>f_i
\label{E:NewObjGreater}
\end{align}
holds, the set of selected RRHs is updated as $\mathcal{A}_{n,i+1}=\mathcal{A}_{n,i}\cup\{j_i\}$, and the current maximum objective value is updated as 
\begin{align}
f_{i+1}=f\big(\mathcal{A}_{n,i}\cup\{j_i\},p_{\hat{k}_n,n}(\mathcal{A}_{n,i}\cup\{j_i\})\big).\label{E:UpdCObj}
\end{align} This procedure is continued until no RRH $j_i$ can be found which satisfies~\eqref{E:NewObjGreater}, or there is no more remaining RRH to be searched, i.e. $i=M$. If the algorithm stops at iteration $i$, the final suboptimal RRH selection and corresponding user power allocation are given by $\check{\mathcal{A}}_n=\mathcal{A}_{n,i}$ and $\check{p}_{\hat{k}_n,n}=p_{\hat{k}_n,n}(\mathcal{A}_{n,i})$. Notice that the greedy algorithm would recover the optimal solution to problem~\eqref{P:DFnJDFixUA} whenever the optimal RRH selection consists of at most two RRHs. However, in general, for given dual variables $\bm\lambda$ and $\bm\mu$, the greedy algorithm gives only a suboptimal solution $\check{\bm\alpha}_n,\check{\bm p}_n$ to problem~\eqref{P:DFnJDFixUA}. An outline of the algorithm is given in~\cref{A:GreedyRRHSel}. 
\begin{table}[h]
\centering
\caption{Greedy algorithm for problem~\eqref{P:DFnJDFixUA}}\label{A:GreedyRRHSel}
\begin{framed}
\begin{algorithmic}[1]
\State Initialization: Iteration $i=1$, set of selected RRHs $\mathcal{A}_{n,1}=\emptyset$, maximum objective value $f_1=0$
\For {each $i=1,\dotsc,M$}
\State Find the RRH $j_{i}\in\mathcal{M}\setminus\mathcal{A}_{n,i}$, according to~\eqref{E:MaxRRHji}
\If {$j_i$ satisfies condition~\eqref{E:NewObjGreater}}
\State Update selected set as $\mathcal{A}_{n,i+1}=\mathcal{A}_{n,i}\cup\left\{j_{i}\right\}$
\State Update current maximum objective value $f_{i+1}$ according to~\cref{E:UpdCObj} 
\Else 
\State Stop and return RRH set $\check{\mathcal{A}}_n=\mathcal{A}_{n,i}$ and corresponding user power allocation $\check{p}_{\hat{k}_n,n}=p_{\hat{k}_n,n}(\mathcal{A}_{n,i})$
\EndIf
\EndFor
\end{algorithmic}
\end{framed}
\end{table}

From~\eqref{E:MaxRRHji}, it can be seen that in each iteration $i=1,\dotsc,M$, the greedy algorithm searches over the set of RRHs $\mathcal{M}\setminus\mathcal{A}_{n,i}$, which has size $M-(i-1)$ to find the RRH $j_i$ according to~\eqref{E:MaxRRHji}. Since there can be at most $M$ iterations, the worst-case complexity of the greedy algorithm is $\sum_{i=1}^{M}M-i+1=M(M+1)/2$, which 
is $O(M^2)$. Since the complexity of finding the user power allocation in each iteration does not depend on $M$, the overall complexity of using the greedy algorithm to solve problem~\eqref{P:DFnJDFixUA} on each SC $n$ is $O(M^2)$, which is only quadratic in $M$, compared to the exponential complexity of the exhaustive search over all possible RRH selections. Thus, in the algorithm in~\cref{A:Overall}, if problem~\eqref{P:DFnJDFixUA} is solved using the suboptimal greedy algorithm in~\cref{A:GreedyRRHSel} instead of the optimal exhaustive search, the worst-case complexity of the algorithm in~\cref{A:Overall} is reduced to $O(NK(M^2+M))$. However, using the greedy algorithm, a convergence to the optimal dual variables $\bm\lambda^\star$ and $\bm\mu^\star$ cannot be guaranteed and the primal solution obtained for problem~\eqref{P:ULHybMain} by this method need not always be feasible. In this case, the power allocations can be made feasible by scaling each of the power constraints in~\eqref{C:PCULHybMain}. Similarly, the constraint in~\eqref{E:IFHRC} can be made feasible by considering each SC in increasing order of the rates achieved. If the DaF mode was selected on an SC, we make the user power allocation zero on this SC. Otherwise, if the FaD mode was selected with more than one RRH processing the signal in the SC, the RRHs are de-selected in increasing order of their contribution to the SNR on each SC, until~\eqref{E:IFHRC} is satisfied. Thus, a feasible solution to problem~\eqref{P:ULHybMain} can always be obtained. However, in~\cref{Sec:SimResults}, it is shown through extensive simulations that there is only a negligible difference between the performance of the greedy algorithm 
and the exhaustive search for $\bm\alpha_n$'s in several practical scenarios. 
Thus, in practice, the joint resource allocation problem~\eqref{P:ULHybMain} may be solved close to optimally at a worst-case complexity of $O(NK(M^2+M))$ by using the greedy algorithm, which is a large reduction compared to a complexity of $O(NK(2^M+M))$ for the optimal algorithm. 
\section{Simulation Results}\label{Sec:SimResults}
For the simulation setup, we first consider a UD-CRAN cluster with $M=5$ RRHs and $K=3$ users. One RRH is located in the center of a square region with side 375 meters~(m), while the others are located on the vertices. The users are randomly located within a larger square region of side 750~m, and whose center coincides with that of the RRH square region. The fronthaul links of all the RRHs are assumed to have the same capacity $\bar{R}_m=\bar{R}~\forall m\in\mathcal{M}$, and all the RRHs use the same resolution of $\beta_m=\beta,~\forall m\in\mathcal{M}$ bits for the uniform SQ in the FaD processing mode. The wireless channel is centered at a frequency of $2$~GHz with a bandwidth $B=20$~MHz, following the Third Generation Partnership Project~(3GPP) Long Term Evolution-Advanced~(LTE-A) standard~\cite{3gpp36211}, and is divided into $N=64$ SCs using OFDMA. The combined path loss and shadowing is modeled as $38+30\log_{10}(d_{m,k})+X$ in dB~\cite{3gpp36931}, where $d_{m,k}$ is the distance between RRH $m$ and user $k$ in meters, and $X$ is the shadowing random variable, which is normally distributed with a standard deviation of $6$~dB. 
The multi-path on each wireless 
channel is modeled using an exponential power delay profile with $N/4$ taps 
and the small-scale fading on each tap is assumed to follow the Rayleigh distribution. The AWGN is assumed to have a power spectral density of $-174$~dBm/Hz with an additional noise figure of $6$~dB at each RRH, while the total transmit power at each user is $\bar{P}_k=23~\text{dBm}\,\forall k\in\mathcal{K}$, 
unless mentioned otherwise. The proposed algorithm for hybrid decoding is compared to the following two benchmark schemes:
\begin{itemize}
\item \textbf{FaD processing on all SCs}: In this case, 
we solve problem~\eqref{P:ULHybMain} with $\delta_n=0\ \forall n\in\mathcal{N}$ using a similar algorithm as given in~\cref{A:Overall}.
\item \textbf{DaF processing on all SCs}: In this case, 
we solve problem~\eqref{P:ULHybMain} with $\delta_n=1\ \forall n\in\mathcal{N}$. 
\end{itemize}
For simplicity, we consider maximization of the sum rate in problem~\eqref{P:ULHybMain}, i.e., the user rate weights $\omega_k=1~\forall k\in\mathcal{K}$, and the values averaged over many random user locations and channels are plotted against various system parameters. 

\begin{figure}[h]
\centering
\includegraphics[width=\linewidth]{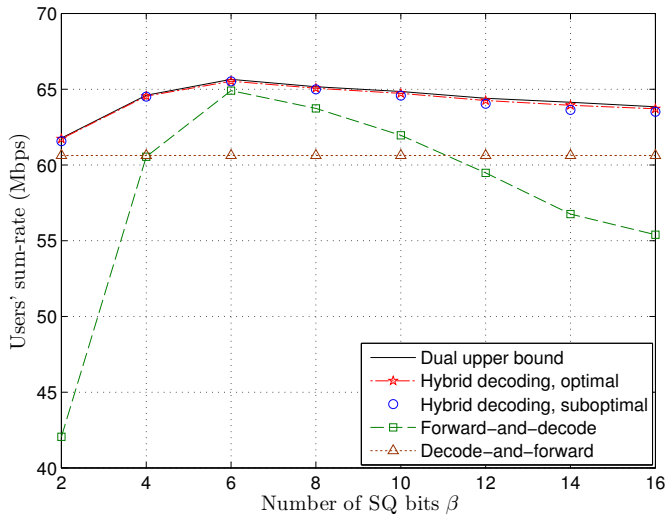}
\caption{Users' sum-rate vs.\ number of SQ bits for system with $M=5$, $K=3$, $B=20$~MHz, $N=64$ and $\bar{R}=250$~Mbps.}\label{F:SRvQ}
\end{figure}
\cref{F:SRvQ} plots the sum-rate against the number of bits $\beta$ used for the SQ at each RRH, when the common fronthaul capacity of the RRHs is $\bar{R}=250$~Mbps. From~\cref{F:SRvQ}, it is observed that the hybrid decoding with the proposed optimal and suboptimal algorithms outperforms both the benchmark schemes described above. Moreover, there is a negligible difference between the optimal and suboptimal hybrid decoding schemes, which implies that in practical UD-CRAN systems, the hybrid decoding gains can be achieved with the suboptimal greedy algorithm at a much lower complexity compared to the optimal algorithm. At lower values of $\beta$, the FaD scheme performs worse than the DaF scheme, since in this case the SQ is coarse, which reduces the achievable rate on each SC as given by~\eqref{E:SCRJ}, even though the fronthaul is able to support this rate. On the other hand, when $\beta$ is very high, the higher achievable rates given by~\eqref{E:SCRJ} cannot be supported by the limited fronthaul capacity of the RRHs, which again restricts the performance of the FaD scheme. However, for values of $\beta$ in between these two extremes, the FaD scheme outperforms the DaF scheme, while the proposed hybrid decoding offers the maximum advantage for all values of $\beta$. 
\begin{figure}[h]
\centering
\includegraphics[width=\linewidth]{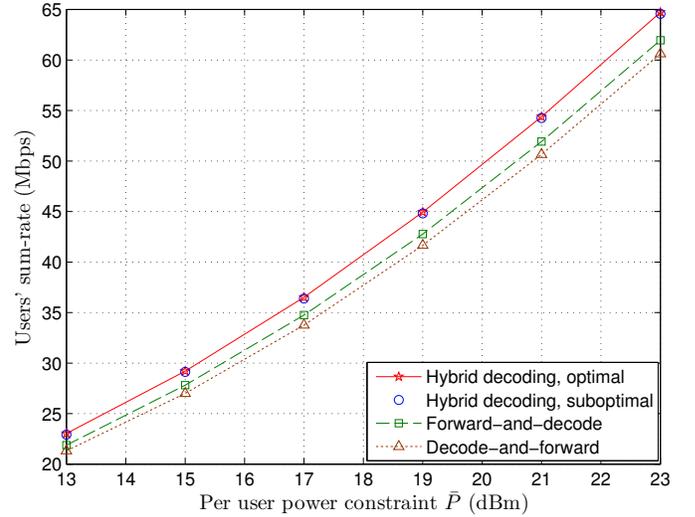}
\caption{Users' sum-rate vs.\ per user total transmit power constraint for system with $M=5$, $K=3$, $B=20$~MHz, $N=64$, $\bar{R}=250$~Mbps and $\beta=10$.}\label{F:SRvP}
\end{figure}
\cref{F:SRvP} plots the sum-rate against the maximum transmit power constraint $\bar{P}$ at each user, when the common fronthaul capacity of the RRHs is $\bar{R}=250$~Mbps, and $\beta=10$. Again, from~\cref{F:SRvP}, it is observed that the proposed hybrid decoding outperforms both the benchmark schemes above, while the performance of the proposed optimal and suboptimal algorithms are similar. In addition to the proposed algorithms and the benchmark schemes listed above,~\cref{F:SRvQ,F:SRvP} also plot the dual upper bound to problem~\eqref{P:ULHybMain} given by the optimal value of the dual function $g(\bm\lambda^\star,\bm\mu^\star)$. It can be observed that the difference between the proposed optimal and suboptimal solutions and the dual upper bound is negligible even for $N=64$. Since the duality gap diminishes with $N$, this implies that the proposed algorithms are nearly optimal for practical values of $N$. 

Next, we consider a large network with $M=125$ RRHs whose locations are fixed as shown in~\cref{F:MU125M120KL}, and $K=120$ users randomly located within a square region of side $2$~km. Since the users that are far away from an RRH do not contribute to its achievable rate, to reduce the complexity of our proposed algorithms, we assume that the network is divided into $25$ square clusters of $5$ RRHs each, as shown in~\cref{F:MU125M120KL}. We assume that proper frequency reuse has been assigned over adjacent clusters and thus neglect the inter-cluster interference for simplicity. Then, the algorithms in~\cref{A:Overall,A:GreedyRRHSel} can be executed in parallel in each cluster. \cref{F:SRvFHR} plots the sum-rate of the users, averaged over random channel realizations, against the common fronthaul capacity of the RRHs, with 
$\beta=10$ bits. 
When the fronthaul capacity is low, the DaF scheme performs better than the FaD scheme, and its performance matches that of the proposed hybrid decoding, since in this case the fronthaul cannot support the transmission of accurately quantized signals from multiple RRHs in any SC. On the other hand, when the fronthaul capacity is sufficiently large, FaD outperforms DaF, and its  sum-rate 
approaches that achieved by hybrid decoding. In this case, most or all the RRHs can participate in the joint decoding on each SC, which provides a joint signal processing gain. In between these two extremes, the proposed hybrid decoding offers significant gains over both the DaF and FaD schemes. In this regime, by performing DaF processing on some of the SCs, the hybrid decoding saves the fronthaul capacities of the respective RRHs, which can in turn be used for carrying quantized signals for FaD processing on other SCs, thus performing better than both the benchmarks. Thus, under practical setups with finite fronthaul capacities, even the suboptimal hybrid decoding algorithm with a low complexity of $O(NK(M^2+M))$ can offer significant throughput gains over the conventional CRAN with FaD processing at all SCs, as well as a cellular network with DaF processing at all SCs. 
\begin{figure}[h]
\centering
\includegraphics[width=\linewidth]{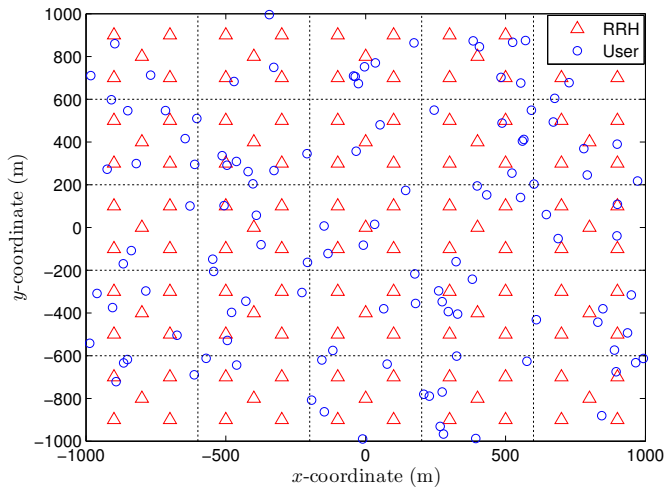}
\caption{Example UD-CRAN layout with $M=125$ and $K=120$.}\label{F:MU125M120KL}
\end{figure}
\begin{figure}[h]
\centering
\includegraphics[width=\linewidth]{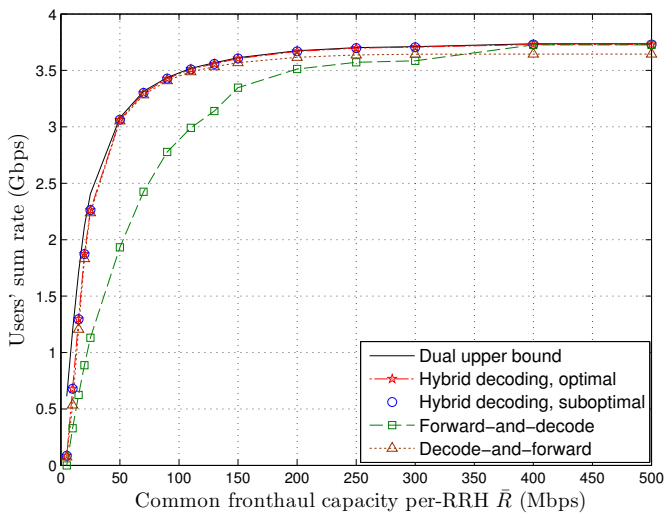}
\caption{Users' sum-rate vs.\ per-RRH fronthaul capacity for system with $M=125$, $K=120$, $B=20$~MHz, $N=64$ and $\beta=10$.}\label{F:SRvFHR}
\end{figure}

\section{Conclusion}\label{Sec:Conc}
In this paper, we have studied the uplink transmission in an OFDMA-based UD-CRAN with hybrid decoding at the RRHs. We formulated a joint RRHs' processing mode selection, user-SC assignment, and users' power allocation problem to maximize the weighted-sum-rate of the users over all SCs subject to the given RRHs' individual fronthaul capacity constraints and the individual transmit power constraints at the users. Although the problem is non-convex, we propose two efficient solutions based on the Lagrange duality technique. Through numerical simulations, it is shown that the proposed algorithms for hybrid decoding with optimized resource allocation outperform both state-of-the-art CRAN with FaD processing and conventional cellular network with DaF processing. 
\appendices
\section{Proof of~\cref{L:rknQConc}}\label{App:ProofLrknQConc}
In this proof, we drop the user and SC sub-scripts $k$ and $n$ for convenience. 
Then, by direct differentiation, it can be shown that the second-order derivative of $\gamma^Q_m(p)$ in~\eqref{E:PartSNR} with respect to~(w.r.t.) $p$ satisfies $\od[2]{\gamma^Q_m(p)}{p}\leq 0,\,\forall p\geq 0$, which implies that $\gamma^Q_m(p)$ in~\eqref{E:PartSNR} is 
concave for $p\geq 0$. 

For a given $\bm\alpha$, since $\gamma^Q(\bm\alpha,p)$ in~\eqref{E:RxSNRJD} is the non-negative sum of the concave functions $\gamma^Q_m(p),\,m\in\mathcal{A}_n$, it follows that $\gamma^Q(\bm\alpha,p)$ is also concave in $p$~\cite{boyd2004convex}. Now, the logarithm function is concave and its extended value extension on the real line is non-decreasing. Thus, for given $\bm\alpha$, $r^Q(\bm\alpha,p)=(B/N)\log_2(1+\gamma^Q(\bm\alpha,p))$ according to~\eqref{E:SCRJ} is the composition of the concave function $\gamma^Q(\bm\alpha,p)$ with a concave and non-decreasing function, and hence $r^Q\left(\bm\alpha,p\right)$ is also concave for $p\geq 0$~\cite{boyd2004convex}, which completes the proof. 
\section{Proof of~\cref{Prop:OptPASRRHQ}}\label{App:ProofOptPASRRHQ}
In this proof, we drop the user and SC sub-scripts 
for convenience. 
Since the objective of problem~\eqref{P:DFnJDFUARRHSel} is concave in $p$, the user power allocation that maximizes its value can be found by differentiating 
the objective w.r.t. $p$ and setting the derivative equal to zero; this leads to the equation, 
\begin{align}
&\frac{-\mu\theta_{\tilde{m}_n}|h_{\tilde{m}_n}|^2}{\sigma_{\tilde{m}_n}^2}p^2-\frac{\mu(2\theta_{\tilde{m}_n}+1)|h_{\tilde{m}_n}|^2}{\sigma_{\tilde{m}_n}^2}p
-\mu(\theta_{\tilde{m}_n}+1)\notag\\
&+\frac{\omega B|h_{\tilde{m}_n}|^2}{\sigma_{\tilde{m}_n}^2N\ln 2}=0.\label{E:DiffPASRRHQ}
\end{align}
Further, it can be observed that the user power allocation $\tilde{p}$ that maximizes the objective of problem~\eqref{P:DFnJDFUARRHSel} is given by the right-hand root of the quadratic equation~\eqref{E:DiffPASRRHQ}, which can be expressed as 
\begin{align}
\tilde{p}&=\frac{(\theta_{\tilde{m}_n}+2)\sigma_{\tilde{m}_n}^2}{2|h_{\tilde{m}_n}|^2}\Bigg(\bigg(1+\frac{4}{\left(\theta_{\tilde{m}_n}+2\right)^2}\notag\\
&\quad\cdot\Big(\frac{\omega B|h_{\tilde{m}_n}|^2\theta_{\tilde{m}_n}}{\sigma_{\tilde{m}_n}^2\mu N\ln 2}-(\theta_{\tilde{m}_n}+1)\Big)\bigg)^{\frac{1}{2}}-1\Bigg).\label{E:OptPAJDProof}
\end{align}
Also, since we require $\tilde{p}\geq 0$, from~\eqref{E:OptPAJDProof} we have 
\begin{align}
\frac{\omega B|h_{\tilde{m}_n}|^2\theta_{\tilde{m}_n}}{\sigma_{\tilde{m}_n}^2\mu N\ln 2}-(\theta_{\tilde{m}_n}+1)\geq 0.\label{E:CondOptPAJDProof}
\end{align}
Combining~\eqref{E:CondOptPAJDProof} with~\eqref{E:OptPAJDProof} gives the expression in~\eqref{E:OptPASingleRRH}, which completes the proof. 
\bibliographystyle{IEEEtran_mod}
\bibliography{IEEEabrv,bibJournalList,ThesisBibliography}

\begin{thebibliography}{10}
\providecommand{\url}[1]{#1}
\csname url@samestyle\endcsname
\providecommand{\newblock}{\relax}
\providecommand{\bibinfo}[2]{#2}
\providecommand{\BIBentrySTDinterwordspacing}{\spaceskip=0pt\relax}
\providecommand{\BIBentryALTinterwordstretchfactor}{4}
\providecommand{\BIBentryALTinterwordspacing}{\spaceskip=\fontdimen2\font plus
\BIBentryALTinterwordstretchfactor\fontdimen3\font minus
  \fontdimen4\font\relax}
\providecommand{\BIBforeignlanguage}[2]{{%
\expandafter\ifx\csname l@#1\endcsname\relax
\typeout{** WARNING: IEEEtran.bst: No hyphenation pattern has been}%
\typeout{** loaded for the language `#1'. Using the pattern for}%
\typeout{** the default language instead.}%
\else
\language=\csname l@#1\endcsname
\fi
#2}}
\providecommand{\BIBdecl}{\relax}
\BIBdecl

\bibitem{andrews-etal2014what}
J.~Andrews \emph{et~al.}, ``What will {5G} be?'' \emph{{IEEE} J. Sel. Areas
  Commun.}, vol.~32, no.~6, pp. 1065--1082, June 2014.

\bibitem{bhushan-etal2014network}
N.~Bhushan \emph{et~al.}, ``Network densification: the dominant theme for
  wireless evolution into {5G},'' \emph{{IEEE} Commun. Mag.}, vol.~52, no.~2,
  pp. 82--89, Feb. 2014.

\bibitem{park-etal2013joint}
S.-H. Park, O.~Simeone, O.~Sahin, and S.~Shamai, ``Joint precoding and
  multivariate backhaul compression for the downlink of cloud radio access
  networks,'' \emph{IEEE Trans.\ Signal Process.}, vol.~61, no.~22, pp.
  5646--5658, Nov. 2013.

\bibitem{zhou-yu2014optimized}
Y.~Zhou and W.~Yu, ``Optimized backhaul compression for uplink cloud radio
  access network,'' \emph{IEEE {J.} Sel.\ Areas Commun.}, vol.~32, no.~6, pp.
  1295--1307, June 2014.

\bibitem{liu-zhang2015optimized}
L.~Liu and R.~Zhang, ``Optimized uplink transmission in multi-antenna {C-RAN}
  with spatial compression and forward,'' \emph{{IEEE} Trans. Signal Process.},
  vol.~63, no.~19, pp. 5083--5095, Oct. 2015.

\bibitem{jain-etal2016backhaul}
S.~Jain, S.~Kim, and G.~Giannakis, ``Backhaul-constrained multicell cooperation
  leveraging sparsity and spectral clustering,'' \emph{{IEEE} Trans. Wireless
  Commun.}, vol.~15, no.~2, pp. 899--912, Feb. 2016.

\bibitem{zhuang-lau2014backhaul}
F.~Zhuang and V.~Lau, ``Backhaul limited asymmetric cooperation for {MIMO}
  cellular networks via semidefinite relaxation,'' \emph{{IEEE} Trans. Signal
  Process.}, vol.~62, no.~3, pp. 684--693, Feb. 2014.

\bibitem{dai-yu2014sparse}
B.~Dai and W.~Yu, ``Sparse beamforming and user-centric clustering for downlink
  cloud radio access network,'' \emph{IEEE Access}, vol.~2, pp. 1326--1339,
  2014.

\bibitem{shi-etal2014group}
Y.~Shi, J.~Zhang, and K.~Letaief, ``Group sparse beamforming for green
  cloud-{RAN},'' \emph{{IEEE} Trans. Wireless Commun.}, vol.~13, no.~5, pp.
  2809--2823, May 2014.

\bibitem{fan-etal2016randomized}
C.~Fan, X.~Yuan, and Y.~J.~A. Zhang, ``Randomized {G}aussian message passing
  for scalable uplink signal processing in {C-RANs},'' in \emph{Proc.\ {IEEE
  Intl.\ Conf.\ Commun.~(ICC)}}, May 2016, pp. 1--6.

\bibitem{luo-etal2015downlink}
S.~Luo, R.~Zhang, and T.~J. Lim, ``Downlink and uplink energy minimization
  through user association and beamforming in {C-RAN},'' \emph{{IEEE} Trans.
  Wireless Commun.}, vol.~14, no.~1, pp. 494--508, Jan. 2015.

\bibitem{tao-etal2016content}
M.~Tao, E.~Chen, H.~Zhou, and W.~Yu, ``Content-centric sparse multicast
  beamforming for cache-enabled cloud {RAN},'' \emph{{IEEE} Trans. Wireless
  Commun.}, vol.~15, no.~9, pp. 6118--6131, Sep. 2016.

\bibitem{shi-etal2016smoothed}
Y.~Shi \emph{et~al.}, ``Smoothed {$L_p$}-minimization for green cloud-{RAN}
  with user admission control,'' \emph{{IEEE} J. Sel. Areas Commun.}, vol.~34,
  no.~4, pp. 1022--1036, Apr. 2016.

\bibitem{fan-etal2016dynamic}
C.~Fan, Y.~J. Zhang, and X.~Yuan, ``Dynamic nested clustering for parallel
  {PHY}-layer processing in cloud-{RAN}s,'' \emph{{IEEE} Trans. Wireless
  Commun.}, vol.~15, no.~3, pp. 1881--1894, Mar. 2016.

\bibitem{peng-etal2016recent}
M.~Peng, Y.~Sun, X.~Li, Z.~Mao, and C.~Wang, ``Recent advances in cloud radio
  access networks: System architectures, key techniques, and open issues,''
  \emph{{IEEE} Commun. Surveys Tuts.}, vol.~18, no.~3, pp. 2282--2308,
  third-quarter 2016.

\bibitem{liu-etal2015joint}
L.~Liu, S.~Bi, and R.~Zhang, ``Joint power control and fronthaul rate
  allocation for throughput maximization in {OFDMA}-based cloud radio access
  network,'' \emph{{IEEE} Trans. Commun.}, vol.~63, no.~11, pp. 4097--4110,
  Nov. 2015.

\bibitem{stephen-zhang2016green}
R.~G. Stephen and R.~Zhang, ``Green {OFDMA} resource allocation in
  cache-enabled {CRAN},'' in \emph{Proc.\ IEEE Online Conf.\ Green
  Commun.~(OnlineGreenComm)}, Nov. 2016, pp. 70--75.

\bibitem{shi-etal2015large}
Y.~Shi, J.~Zhang, K.~Letaief, B.~Bai, and W.~Chen, ``Large-scale convex
  optimization for ultra-dense cloud-{RAN},'' \emph{{IEEE} Wireless Commun.
  Mag.}, vol.~22, no.~3, pp. 84--91, June 2015.

\bibitem{stephen-zhang2017joint}
R.~G. Stephen and R.~Zhang, ``Joint millimeter-wave fronthaul and {OFDMA}
  resource allocation in ultra-dense {CRAN},'' \emph{{IEEE} Trans. Commun.},
  vol.~65, no.~3, pp. 1411--1423, Mar. 2017.

\bibitem{bi-etal2015wireless}
S.~Bi, R.~Zhang, Z.~Ding, and S.~Cui, ``Wireless communications in the era of
  big data,'' \emph{{IEEE} Commun. Mag.}, vol.~53, no.~10, pp. 190--199, Oct.
  2015.

\bibitem{patil-yu2014hybrid}
P.~Patil and W.~Yu, ``Hybrid compression and message-sharing strategy for the
  downlink cloud radio-access network,'' in \emph{Proc.\ Inform.\ Theory and
  Appl.\ Workshop~(ITA)}, Feb. 2014, pp. 1--6.

\bibitem{liu-yu2017cross}
L.~Liu and W.~Yu, ``Cross-layer design for downlink multihop cloud radio access
  networks with network coding,'' \emph{{IEEE} Trans. Signal Process.},
  vol.~65, no.~7, pp. 1728--1740, Apr. 2017.

\bibitem{yu-lui2006dual}
W.~Yu and R.~Lui, ``Dual methods for nonconvex spectrum optimization of
  multicarrier systems,'' \emph{{IEEE} Trans. Commun.}, vol.~54, no.~7, pp.
  1310--1322, July 2006.

\bibitem{goldsmith2005wireless}
A.~Goldsmith, \emph{Wireless Communications}.\hskip 1em plus 0.5em minus
  0.4em\relax Cambridge University Press, 2005.

\bibitem{boyd2014ellipsoid}
\BIBentryALTinterwordspacing
S.~Boyd, ``Ellipsoid method: Notes for {EE364b}, {S}tanford {U}niversity,'' May
  2014. [Online]. Available:
  \url{http://stanford.edu/class/ee364b/lectures/ellipsoid_method_notes.pdf}
\BIBentrySTDinterwordspacing

\bibitem{3gpp36211}
\emph{Evolved Universal Terrestrial Radio Access~({E-UTRA}); Physical channels
  and modulation~(Release 12)}, 3GPP Std. 36.211, 2014.

\bibitem{3gpp36931}
\emph{Evolved Universal Terrestrial Radio Access~({E-UTRA}); Radio
  Frequency~({RF}) requirements for {LTE Pico Node B}~(Release 12)}, 3GPP Std.
  36.931, 2014.

\bibitem{boyd2004convex}
S.~Boyd and L.~Vandenberghe, \emph{Convex Optimization}.\hskip 1em plus 0.5em
  minus 0.4em\relax Cambridge University Press, 2004.

\end{thebibliography}
\end{document}